\documentclass[a4paper,10pt]{article}

\usepackage[utf8x]{inputenc}
\usepackage{amsmath,amssymb,amsthm}
\usepackage[english]{babel}

\usepackage{graphicx}
\usepackage{caption}
\usepackage{subcaption}

\usepackage{enumitem}
\newtheorem{thm}{Theorem}[section]
\newtheorem{lem}[thm]{Lemma}
\newtheorem{defn}[thm]{Definition}

\newtheorem{claim}{Claim}[thm]
     
\newcommand{\first}{\operatorname{first}}
\newcommand{\last}{\operatorname{last}}
\newcommand{\previous}{\operatorname{previous}}
\newcommand{\next}{\operatorname{next}}

\newcommand{\Lp}{\operatorname{LeftPoint}}
\newcommand{\Rp}{\operatorname{RightPoint}}

\title{Drawing the Almost Convex Set in an Integer Grid of Minimum Size}
\author{ Frank Duque \thanks{Departamento de Matem\'aticas, CINVESTAV. 
 } \and
 Ruy Fabila-Monroy \footnotemark[1] \and
 Carlos Hidalgo-Toscano \footnotemark[1] \and
 Pablo Pérez-Lantero \thanks{Departamento de Matemática y Ciencia de la Computación, Universidad de Santiago, Chile.}}

\begin{document}

\maketitle

% \linenumbers

\begin{abstract} 
In 2001, Károlyi, Pach and Tóth  introduced 
a family of point sets to solve an Erd\H{o}s-Szekeres type problem; 
which have been used to solve several other Ed\H{o}s-Szekeres type problems.
In this paper we refer to these sets as nested almost convex sets. 
A nested almost convex set  $\mathcal{X}$ has the property that the interior of every triangle 
determined by three points in the same convex layer of $\mathcal{X}$,
 contains exactly one point of $\mathcal{X}$. In this paper, we introduce a characterization of nested
 almost convex sets. Our characterization implies that there exists at most one (up to order type)
 nested almost convex set of $n$ points.
 We use our characterization to obtain a linear time algorithm to construct nested almost convex sets of $n$ points, 
 with integer coordinates of absolute values at most $O(n^{\log_2 5})$. 
 Finally,  we use our characterization to obtain an
 $O(n\log n)$-time algorithm to determine whether a set of points is a nested almost convex set.
\end{abstract}

\section{Introduction.} 
We say that a set of points in the plane is in \emph{general position} 
if no three of them are collinear.
Throughout this paper all points sets are in general position.
In \cite{Erdos_EmptyConvex_1978}, Erd\H{o}s asked for the minimum integer $E(s,l)$ that satisfies the following. 
Every set of at least $E(s,l)$ points,
contains $s$ points in convex position and at most $l$ points in its interior.
A \emph{$k$-hole} of $\mathcal{X}$ is a polygon with $k$ vertices, all of which belong to $\mathcal{X}$
and has no points of $\mathcal{X}$ in its interior; the polygon
may be convex or non-convex.
In 1983, Horton surprised the community
with a simple proof that $E(s,l)$ does not exist for $l=0$ and $s\geq 7$ \cite{Horton_set_1983};
Horton constructed arbitrarily large point set with no convex $7$-holes.
Note that for $l=0$, $E(s,l)$ is the minimum integer such that 
every set of at least $E(s,0)$ points contains at least one $s$-hole.

In 2001 \cite{Pach_almostConvex1} Károlyi, Pach and Tóth  introduce 
a family of sets that, 
although was not given a name, it was used in other works related to the original question
of Erd\H{o}s.
In this paper we refer the elements of this family as \emph{nested almost convex sets}.
They have been used in the following problems.

\begin{description} [leftmargin=*]
 \item[A modular version of the Erd\H{o}s problem.] 
 In 2001 \cite{Pach_almostConvex1} Károlyi, Pach and Tóth use the nested almost convex sets to prove that,
 for any $s\geq5l/6+O(1)$, there is an integer $B(s,l)$ with the following property.
 Every set of at least $B(s,l)$ points in general position contains $s$ points in convex position 
 such that the number of points in the interior of their convex hull is $0$, modulo $(l)$.
 This "modular" version of the Erd\H{o}s problem was proposed by
 Bialostocki, Dierker, and Voxman \cite{Bialostocki1991}.
 This was proved for $s\geq l+2$ by Bialostocki et al.
 The original upper bound on $B(s,l)$ was later improved by Caro in \cite{Caro1996}.

 \item[A version of the Erd\H{o}s problem in almost convex sets.]
 We say that $\mathcal{X}$ is an \emph{almost convex} set if every triangle
 with vertices in $\mathcal{X}$ contains at most one point of $\mathcal{X}$ in its
 interior.
 Let $N(s)$ be the smallest integer such that every almost convex set of at least $N(s)$ points 
 contains an $s$-hole.
 In 2007 \cite{Valtr2007AlmostConvex2} Valtr Lippner and Károlyi 
 use the nested almost convex sets to prove that:
 \begin{equation} \label{eq:EqualityHolesAlmostConvex}
  N(s)=  \begin{cases}
    2^{(s+1)/2}-1       & \quad \text{if } s\geq 3 \text{ is odd}\\
    \frac{3}{2}2^{s/2}-1  & \quad \text{if } s\geq 4 \text{ is even.}\\
  \end{cases} 
 \end{equation}
 The authors use the nested almost convex sets to attain the equality in (\ref{eq:EqualityHolesAlmostConvex}).
 The existence of $N(s)$ was first proved by Károlyi, Pach and Tóth in \cite{Pach_almostConvex1}.
 The upper bound for $N(s)$ was improved by  Kun and Lippner in \cite{Kun2003}, 
 and it was improved again by Valtr in \cite{Pavel_CupsCaps}.
 
 \item[Maximizing the number of non-convex $4$-holes.] 
 In 2014 \cite{FourHoles} Aichholzer, Fabila-Monroy, Gonz{\'a}lez-Aguilar, 
 Hackl, Heredia, Huemer, Urrutia and Vogtenhuber prove that 
 the maximum number of non-convex $4$-holes in a set of $n$ points is at most $n^3/2-\Theta(n^2)$.
 The authors use the nested almost convex sets to prove that some sets have 
 $n^3/2-\Theta(n^2\log(n) )$ non-convex $4$-holes.
 
 \item[Blocking $5$-holes.]
 A set $B$ blocks the convex $k$-holes in $\mathcal{X}$, 
 if any $k$-hole of $\mathcal{X}$ contains at least one element of $B$ 
 in the interior of its  convex hull.
 In 2015 \cite{Bloking} Cano, Garcia, Hurtado, Sakai, Tejel and Urritia  
 use the nested almost convex sets to prove that:
 $n/2-2$ points are always necessary and sometimes sufficient
 to block the $5$-holes of a point set with $n$ elements in convex position and $n=4k$.
 The authors use the nested almost convex sets as an example of
 a set for which $n/2-2$ points are sufficient to block its $5$-holes.
 
\end{description}
We now define formally the nested almost convex sets.
\begin{defn}\label{def:nested} Let $\mathcal{X}$ be a point set; let 
$k$ be the number of convex layers of $\mathcal{X}$; and for $1\leq j \leq k$, let
$R_j$ be the set of points in the $j$-th convex layer of $\mathcal{X}$.
We say that $\mathcal{X}$ is a nested almost convex set if:
\begin{enumerate}
 \item $\mathcal{X}_j:=R_1\cup R_2 \cup \dots \cup R_j$ is in general position,
 \item the vertices in the convex hull of $\mathcal{X}_j$ are the elements of $R_j$, and
 \item any triangle determined by three points of $R_j$ contains precisely one point of 
 $\mathcal{X}_{j-1}$ in its interior.
\end{enumerate}
\end{defn}

  In this paper, we give a characterization of when a set of points is a nested almost convex set. This is done by first
  defining a family of trees. If there exists a map, that satisfies certain properties, from the point set to the nodes of a tree in the family, then
  the point set is a nested almost convex set. This map encodes a lot of information about the point set. 
  For example, it determines the location of any given  point with respect to the convex hull; we use this information 
  to obtain an  $O(n\log n)$-time algorithm to decide  whether a set of points is a nested almost convex set.
  This map also determines the orientation of any given triplet of points. This implies that for every $n$ there exists
  essentially at most one nested almost convex set. We further apply this information
  to obtain a linear-time algorithm that produces a representation of a nested almost convex set of $n$ points on a small integer grid of size $O(n^{\log_2 5})$.

The \emph{order type} of a point set $\mathcal{X}=\lbrace x_1, x_2, \dots x_n\rbrace$
is a mapping that assigns to each ordered triplet $(x_i,x_j,x_k)$ an orientation.
If $x_k$ is to the left of the directed line from $x_i$ to $x_j$,
the orientation of $(x_i,x_j,x_k)$ is counterclockwise.
If $x_k$ is to the right of the directed line from $x_i$ to $x_j$,
the orientation of $(x_i,x_j,x_k)$ is clockwise.
We say that two set of points have the  same order type, 
if there exist a bijection between these sets that preserves the orientation of all triplets. 

The order type was introduced by Goodman and Pollack in \cite{m_sorting},
and it has been widely used in Combinatorial Geometry to classify point sets; 
two sets of points are essentially the same if they have the same order type.
As a consequence of the characterization of nested almost convex sets presented in Section \ref{Section:Characterization},
we have the following.

 \begin{thm} \label{thm:PossibleAlmosConvex}
  If $n=2^{k-1}-2$ or $n=3\cdot2^{k-1}-2$ there is exactly one order type 
  that correspond to a nested almost convex set with $n$ points;
  for other values of $n$, nested almost convex sets with $n$ points
  do not exist.
 \end{thm}

In previous papers, two constructions of nested almost convex sets have been presented.
The first construction was introduced by Károlyi, Pach and Tóth in \cite{Pach_almostConvex1}.
The second construction was introduced by Valtr, Lippner and Károlyi in \cite{Valtr2007AlmostConvex2}
six years later. 

\begin{description} [leftmargin=*]
 \item[Construction 1:] Let $X_1$ be a set of two points.
  Assume that $j>0$ and that $X_j$ has been constructed.  
  Let $z_1,\dots z_r$ denote the vertices of $R_j$ in clockwise order. 
  Let $P_j$ be the polygon with vertices in $R_j$. Let $\varepsilon_j , \delta_j >0$.
  For any $1\leq i\leq r$, let $\ell_i$ denote the line through $z_i$ orthogonal to the bisector of the angle of $P_j$ at $z_i$. 
  Let $z_i'$ and $z_i''$ be two points in $\ell_i$ at distance $\varepsilon_j$ of $z_i$. 
  Finally, move $z_i'$ and $z_i''$ away from $P_j$ at distance $\delta_j$,
  in the direction orthogonal to $\ell_i$, and denote the resulting points by $u_i'$ and $u_i''$, respectively.
  Let $R_{j+1}= \lbrace u_i',u_i'':i=1\dots r \rbrace $ and  $X_{j+1}=X_j \cup R_{j+1}$.
  It is easy to see that if $\varepsilon_j$ and $\frac{\varepsilon_j}{\delta_j}$ are sufficiently small, 
  then $X_{j+1}$ is an almost convex set. See Figure \ref{fig:type1}.
 
 \item[Construction 2:] Let $X_1$ be a set of one point. 
  Let $R_2$ be a set of three points such that, 
  the point in $X_1$ is in the interior of the triangle determined by $R_2$.
  Let $X_2=X_1\cup R_2$.
  Now recursively, suppose that $X_j$ and $R_j$ have been constructed and
  construct the next convex layer $R_{j+1}$ as in Construction 1. See Figure \ref{fig:type2}.
\end{description}

\begin{figure}
    \centering
    \begin{subfigure}[b]{0.3\textwidth}
        \includegraphics[scale=.48]{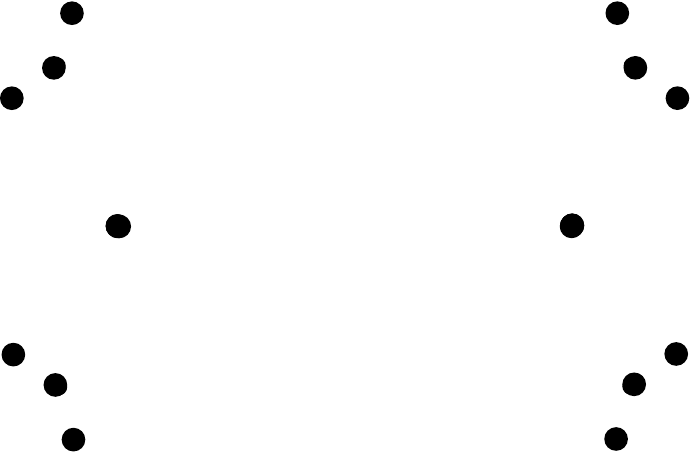}
        \caption{Construction 1}
        \label{fig:type1}
    \end{subfigure}
    \hspace{2 cm } 
    \begin{subfigure}[b]{0.3\textwidth}
        \hspace{5 mm}\includegraphics[scale=.5]{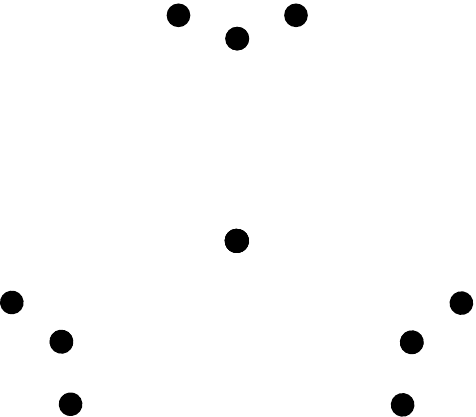}
        \caption{Construction 2}
        \label{fig:type2}
    \end{subfigure}

    \caption{Examples of Almost Convex Sets}\label{fig:types}
\end{figure}

Computers are frequently used to decide whether particular sets satisfy some properties.
Thus, a representation of large nested almost convex sets could be necessary. 
Construction 1 and Construction 2 provide such representations; however,  the coordinates of the points in those constructions
are not integers or are too large with respect to the value of $n$. This is prone to rounding errors or incrases the cost of computation. 
Thus, it is better if the coordinates of the points in the representations
are small integers. 

A \emph{drawing} of $\mathcal{X}$ is a set of points with integer coordinates 
and with the same order type than $\mathcal{X}$.
The \emph{size} of a drawing is the maximum of the absolute values of its coordinates.
Other works on point sets drawings are \cite{DrawingHortonSet, double_circle, goodman,DrawingConvex}.

In Section \ref{section:drawing}, 
we prove that a nested almost convex set  of $n$ points (if it exists), can be drawn in
an integer grid of size $O(n^{\log_2 5}) \simeq O(n^{2.322})$.
Furthermore, we provide a linear time algorithm to find this drawing.
A lower bound of $\Omega(n^{1.5})$ on the size of any drawing
of a nested almost convex set of $n$ points can be derived from the following observations.
Any drawing of an $n$-point set in convex position has size $\Omega(n^{1.5})$ \cite{DrawingConvex}; and every
nested almost convex set of $n$ points has a $\Theta(n)$ points in convex position. This is presented
in detail in Section~\ref{Section:Characterization}.

In Section \ref{section:Certificate}, we are interested in finding an algorithm 
to decide whether a given point set is a nested almost convex set. 
A straightforward $O(n^4)$-time algorithm for this problem can be given using Definition~\ref{def:nested}.
This can be improved to $O(n^2)$ as follows.
Using the algorithm presented in  Section~\ref{section:drawing} an instance
of nested almost convex set can be constructed. Recently in \cite{DecisionOrderType},
Aloupis, Iacono, Langerman, \"Oskan and Wuhrer gave an $O(n^2)$-time algorithm to decide whether two
given sets of $n$ points have the same order type. Thus, using their algorithm and our instance
solves the decision problem in $O(n^2)$ time.
We further improve on this by presenting $O(n \log n)$ time algorithm.

\section{Characterization of Nested Almost Convex \\ Sets.} \label{Section:Characterization}

In this section we prove Theorem \ref{thm:characterization}, 
in which the nested almost convex sets are characterized.
First we introduce some definitions.

% \subsection*{Definitions}\label{SubSection:definitions}

Throughout this section: 
$\mathcal{X}$ will denote a set of $n$ points in general position;
$k$ will denote the number of convex layers of $\mathcal{X}$;
$R_j$ will denote the set of points in the $j$-th convex layer of $\mathcal{X}$, 
$R_1$ being the most internal;
and $\mathcal{X}_j$ will denote the set of points in $\mathcal{X}$,
that are in $R_j$ or in the interior of its convex hull.
\begin{description}
 \item [$\mathbf{T_1(k)}$:] We define $T_1(k)$ as the complete binary tree with $2^{k+1}-1$ nodes.
 The $j$\emph{-level} of $T_1(k)$ is defined as the set of the nodes at distance $j$ from the root.

 \item [Type 1:] We say that $\mathcal{X}$ is of \emph{type 1} if 
 $|R_j|=2^{j}$ for $1\leq j \leq k-1$.
 Note that if $\mathcal{X}$ is of Type 1, then for every $1\leq j \leq k$, 
 the number of points in $R_j$ is equal to the number of nodes in the $j$-level of $T_1(k)$.

 \item [Type 1 labeling:]  An injective function $\psi:\mathcal{X}\rightarrow T_1(k)$ is a \emph{type 1 labeling}, 
 if $\mathcal{X}$ is Type 1 and $\psi$ labels the nodes (different to the root) of $T_1(k)$
 with different points of $\mathcal{X}$.

 \item [$\mathbf{T_2(k)}$:] We define $T_2(k)$ as the tree that, its root has three children,
 and each child is the root of a complete binary tree with $2^{k-1}-1$ nodes.
 The $j$\emph{-level} of $T_2(k)$ is defined as the set of the nodes at distance $j-1$ from the root.
 
 \item[Type 2:] We say that $\mathcal{X}$ is of \emph{type 2} if 
 $|R_1|=1$ and $|R_j|=3\cdot 2^{j-2}$ for $2\leq j \leq k$.
 Note that if $\mathcal{X}$ is of Type 2, the for every $1\leq j \leq k$, 
 the number of points in $R_j$ is equal to the number of nodes in the $j$-level of $T_2(k)$.
 
 \item[Type 2 labeling:] An injective function $\psi:\mathcal{X}\rightarrow T_2(k)$ is a \emph{Type 2 labeling}, 
 if $\mathcal{X}$ is Type 2 and $\psi$ labels the nodes (also the root) of $T_2(k)$ 
 with different points of $\mathcal{X}$.
 
 \item[Labeling:] Let $T$ be equal to $T_1(k)$ or $T_2(k)$.
 We say that a map $\psi:\mathcal{X}\rightarrow T$ is a \emph{labeling},
 if $\psi$ is a Type 1 labeling or a Type 2 labeling. 
 Note that, if $\mathcal{X}$ admits a labeling then $n=2^{k-1}-2$ or $n=3\cdot2^{k-1}-2$.
\end{description}
In the following, when the map $\psi:\mathcal{X}\rightarrow T$ is clear from the context, 
we say that a point is the \emph{label} of a node of $T$
if the point is mapped to the node by $\psi$.  
This way, given a node $u$ of $T$, we denote by $x_u$ its label.
We denote by $u(l)$ and $u(r)$ the left and right children of $u$ in $T$, respectively.
 
\begin{description}
 \item[Nested:] We say that a labeling is nested if, for $1\leq j \leq k$, 
 the left to right order of labels of the nodes in the $j$-level of $T$,
 corresponds to the counterclockwise order of the points in $R_j$.
 
 \item[Adoptable:]  Given a point $p$ in $R_j$ and two points $q_1,q_2$ in $R_{j+1}$,
 we say that $q_1,q_2$ are \emph{adoptable from} $p$ if, for every other point $q_3$ in $R_{j+1}$,
 $p$ is in the interior of the triangle determined by $q_1$, $q_2$, $q_3$.
 We say that a nested labeling is \emph{adoptable} if, for every node $u$ in $T$, 
 $x_{u(l)}$ and $x_{u(r)}$ are adoptable from $x_u$.
\end{description}
 We denote by $R_j(u)$ the set of points in $R_j$ that label a descendant of $u$.
 With respect to the counterclockwise order, we denote by:
 $\first[R_j(u)]$, the first point in $R_j(u)$; 
 $\last[R_j(u)]$, the last point in $R_j(u)$;
 $\previous[R_j(u)]$, the point in $R_j$ previous to $\first[R_j(u)]$; and
 $\next[R_j(u)]$, the point in $R_j$ next to $\last[R_j(u)]$. See figure \ref{fig:Points-R_j-u}.
\begin{figure}[ht]
	\centering
	\includegraphics[scale=.5]{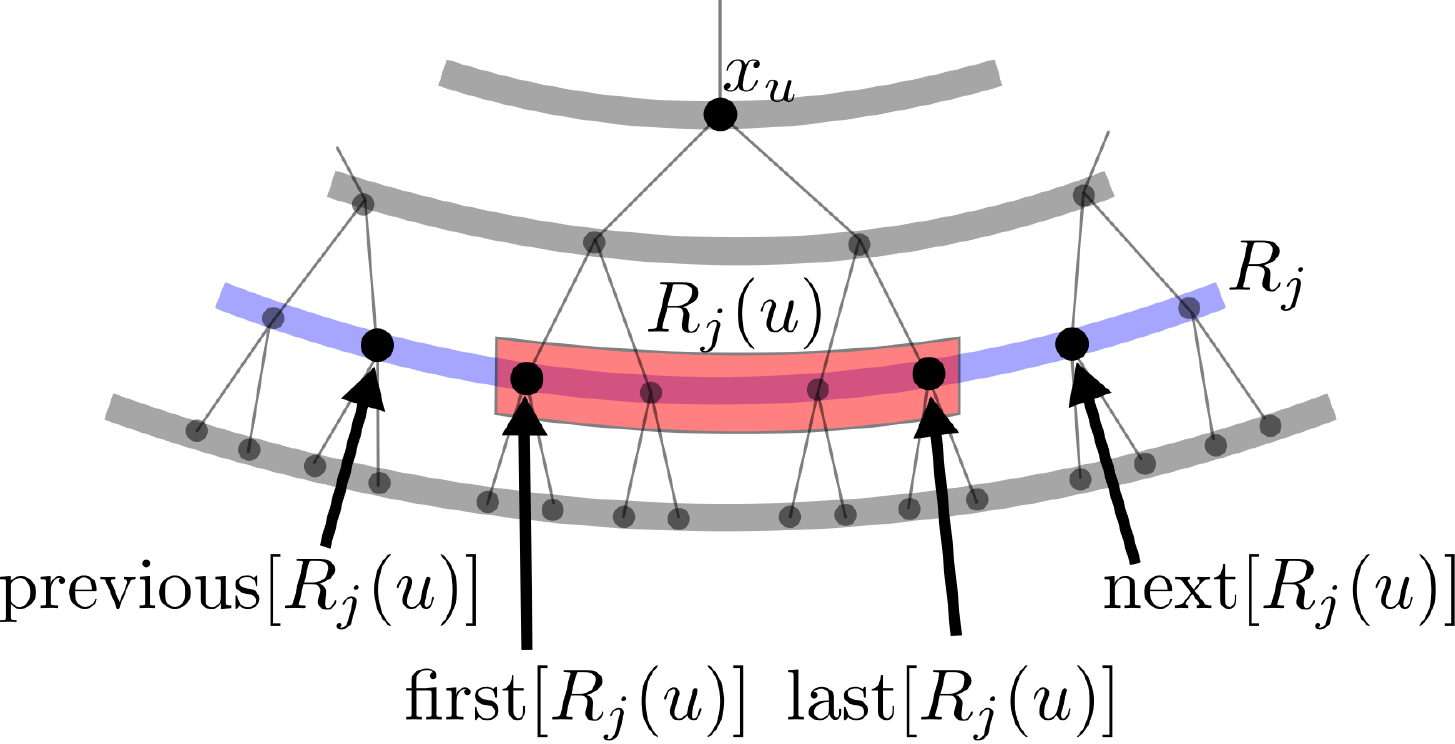}
	\caption{Illustration of $R_j(u)$, $\first[R_j(u)]$, $\last[R_j(u)]$, $\previous[R_j(u)]$, and $\next[R_j(u)]$.}
	\label{fig:Points-R_j-u}
\end{figure}
\begin{description}
 \item[Well laid:] We say that a nested labeling is \emph{well laid} if, for every $u$ in $T$,
  $x_u$ is  in the intersection of the triangle determined  
  by $\previous[R_k(u)]$, $\first[R_k(u)]$, $\last[R_k(u)]$ 
  and the triangle determined by $\first[R_k(u)]$, $\last[R_k(u)]$, \\ $\next[R_k(u)]$.
\end{description}
Let $u$ be a node of $T$. We denote by $\mathcal{X}_u$ the set of points $x_v$ such that
$v$ is descendant of $u$ in $T$.
We denote by $\overline{\mathcal{X}_{u}}$ the set $\mathcal{X}_{u}\cup \{x_u\}$.
Given two sets of points $A$ and $B$, 
we call any directed line from a point in $A$ to a point in $B$, an $(A,B)$-line. 
\begin{description}
 \item[Internal separation:]
 We say that a nested labeling is an internal separation if for every node $u$ of $T$,
 every point in $\mathcal{X}/\mathcal{X}_u$ is to the left of every
 $(\overline{\mathcal{X}_{u(l)}},\overline{\mathcal{X}_{u(r)}})$-line $\ell$.
 
 \item[External separation:]
 We say that a nested labeling is an external separation if for every node $u$ of $T$,
 every point in $\mathcal{X}/\overline{\mathcal{X}_u}$  is to the left of every 
 $(\overline{\mathcal{X}_{u(l)}},\lbrace x_u\rbrace)$-line and to the left of every
 $(\lbrace x_u\rbrace,\overline{\mathcal{X}_{u(r)}})$-line. 
\end{description}

\begin{thm} \label{thm:characterization}
 Let $\mathcal{X}$ be a point set in general position. 
 Then the following statements are equivalent:
 \begin{enumerate}
  \item \label{thm:characterization-def} 
  $\mathcal{X}$ is a nested almost convex set. 
  
  \item \label{thm:characterization-region} 
  $\mathcal{X}$ admits a labeling that is nested, adoptable and well laid. 
  
  \item \label{thm:characterization-order}
  $\mathcal{X}$ admits a labeling that is an internal separation
  and an external separation.
 \end{enumerate}
\end{thm}

\subsection*{Proof of Theorem \ref{thm:characterization}}

The proof of Theorem \ref{thm:characterization} is divided into three parts:
first we prove that $1\implies2$; afterwards we prove that $2\implies3$;
and finally we prove that $3\implies1$.

\subsubsection*{\ref{thm:characterization-def} $\implies$ \ref{thm:characterization-region}} \label{SubSection:AlmostConvexHasLabeling}

In this part we assume that $\mathcal{X}$ is a nested almost convex set, 
and we introduce a labeling $\psi'$  that is nested, adoptable and well laid.

It is clear from the definition of labeling that
a necessary condition for $\mathcal{X}$ to admit a labeling
is that $\mathcal{X}$ must be type 1 or type 2. 
In the following lemma we prove that,
if $\mathcal{X}$ is a nested almost convex, then $\mathcal{X}$ is type 1 or type 2.

 \begin{lem} \label{lem:ConvexLayers}
 If $\mathcal{X}$ is a nested almost convex set then we have one of the following cases:
 \begin{enumerate}
  \item \label{lem:ConvexLayersCase1} $|R_j|=2^{j}$ for $1\leq j \leq k-1$.
  \item \label{lem:ConvexLayersCase2} $|R_1|=1$ and $|R_j|=3\cdot 2^{j-2}$ for $2\leq j \leq k$.
 \end{enumerate}
\end{lem}
\begin{proof}
 Suppose that $R_1$ has three or more points.
 In this case, the interior of the convex hull of $R_1$ has at least one point of $\mathcal{X}$;
 this contradicts that $R_1$ is the first convex layer of $\mathcal{X}$.
 Thus $R_1=\mathcal{X}_1$, and $\mathcal{X}_1$ has one or two points. 
 This proves the lemma for $j=1$.
 
 Any triangulation of $R_{j+1}$,
 has $|R_{j+1}|-2$ triangles  and each triangle has exactly one point of $\mathcal{X}_j$ in its interior;
 thus $|\mathcal{X}_j|=|R_{j+1}|-2$.
 In particular, if $|\mathcal{X}_1|=2$ or $|\mathcal{X}_1|=1$ then $|\mathcal{X}_2|=4$ or $|\mathcal{X}_2|=3$, respectively.
 This proves the lemma for $j=2$.
 
 For the other cases, note that
 \begin{equation*}
  |R_{j+1}|= |\mathcal{X}_j|+2 = |R_{j}|+|\mathcal{X}_{j-1}|+2 = 2|R_{j}|.
 \end{equation*}
\end{proof}  

Now we define $\psi'$ on a subset of nodes of $T$ depending on whether $\mathcal{X}$ is of type 1 or type 2.
 \begin{itemize}
 \item If $\mathcal{X}$ is of type 1: $\psi'$ labels the two nodes in the $1$-level of $T_1(k)$,
 with the two points in $R_1$.
 
 \item If $\mathcal{X}$ is of type 2:
 $\psi'$ labels the node in the $1$-level of $T_2(k)$,
 with the point in $R_1$; $\psi'$ labels the three nodes in the $1$-level of $T_2(k)$,
 with the three points in $R_2$
 (such that, the left to right order of labels of the nodes in the $2$-level of $T$,
 coincides to the counterclockwise order of the points in $R_2$).
 \end{itemize}
To define $\psi'$ on the other nodes of $T$, we use the following Lemma.

 \begin{lem}
 \label{lem:Adoptables}
  Let $p_0,\dots p_{t}$ be the set of points in $R_j$ in counterclockwise order. 
  Then, the points in $R_{j+1}$ can be listed in counterclockwise order as $q_0,q_1,\dots q_{2t+1}$,
  where the points $q_{2i},q_{2i+1}$ are adoptable from $p_i$ for $0\leq i \leq t$.
 \end{lem}
 
 \begin{proof}
  Let $\mathcal{T}$ be the set of triangles 
  determined by three consecutive points of $R_{j+1}$ in counterclockwise order.
  We first show that:
  
  \begin{claim} \label{claim:Adoptables}
   Each point of $R_j$ is in exactly two consecutive triangles of $\mathcal{T}$.
  \end{claim}
  
  Assume that $j\geq 2$ (and note that Claim~\ref{claim:Adoptables} holds for $j=1$).
  Let $\triangle$ be the interior of a triangle of $\mathcal{T}$.
  By the almost convex set definition, there is one point of $\mathcal{X}_j$ in $\triangle$.
  This point must be in $R_j$, 
  since the convex hull of $R_{j+1}$ without $\triangle$ (and its boundary) is convex.
  Thus, there is one point of $R_j$ in the interior of each triangle of $\mathcal{T}$.
  As the triangles of $\mathcal{T}$ are defined by consecutive points of $R_{j+1}$, 
  each point of $R_j$ is in at most two triangles of $\mathcal{T}$.
  Thereby Claim~\ref{claim:Adoptables} follows from $|\mathcal{T}|=|R_{j+1}|=2|R_j|$.

  The two triangles of $\mathcal{T}$ that contain $p_0$, are defined by four consecutive points of $R_{j+1}$;
  let $q_0$ be the second of these points.
  Let $q_0,q_1,\dots q_{2t+1}$ be the points of $R_{j+1}$ in counterclockwise order. 
  Note that, for each $p_i$, 
  the middle two points of the four points that define the two triangles that contain $p_i$,
  are $q_{2i}$ and $q_{2i+1}$.
  Thus $q_{2i}$ and $q_{2i+1}$ are adoptable from $p_i$.
  
 \end{proof}

 Now we define $\psi'$ on the other nodes of $T$ recursively.
 For each labeled node $u$, $\psi'$ labels $u(l)$ and $u(r)$ with the two points 
 adoptable from the label of $u$.
 We do this so that, the left to right order of the labels of the nodes in the $(j+1)$-level of $T$,
 correspond to the counterclockwise order of the points in $R_{j+1}$.
 Note that $\psi'$ is nested and adoptable.
 It remains to prove that $\psi'$ is well laid.
 We prove this in Lemma \ref{lem:LabelInTriangle}.
 
 \begin{lem}
 \label{lem:ConvexHull}
  If $u$ is a node of $T$, 
  the label of every descendant of $u$ is contained in the convex hull of $R_k(u)$.
 \end{lem}

 \begin{proof}
  We claim that
  every set $R_{j-1}(u)$, with at least two points,
  is contained in the convex hull of $R_j(u)$.
  Let $p$ be a point in $R_{j-1}(u)$ and let $q$ and $q'$ be the labels of the children of the node labeled by $p$.
  By construction of $\psi'$, $q$ and $q'$ are adoptable from $p$.
  As $R_{j-1}(u)$ has at least two points, $R_{j}(u)$ has at least four points.
  Let $\triangle$ be a triangle determined by $q$, $q'$ and another point of $R_j(u)$.
  By definition of adoptable, $p$ is in the interior of $\triangle$ 
  and in consequence in the interior of the convex hull of $R_j$.
  An inductive application of the previous claim proves this lemma.
 \end{proof}
 
 \begin{lem}
 \label{lem:LabelInTriangle}
  Let $u$ be a node of $T$. 
  Then $x_u$ is  in the intersection of the triangle determined  
  by $\previous[R_k(u)]$, $\first[R_k(u)]$ and $\last[R_k(u)]$ 
  and the triangle determined by $\first[R_k(u)]$, $\last[R_k(u)]$ and $\next[R_k(u)]$.
 \end{lem}
 \begin{proof}
  Let $j$ be the index such that the $j$-level of $T$ contains $u$.
  Let $R'_k$ be the set that contains  $\first(R_k(v))$ and $\last(R_k(v))$
  for all nodes $v$ in the $j$-level of $T$.
   Let $\mathcal{T}$ be the set of triangles 
   determined by three consecutive points of $R'_{k}$ in counterclockwise order.
   We first show the following claim.  
    \begin{claim}\label{cl:LabelInTriangle}
        Each point of $R_j$ is in exactly two consecutive triangles of $\mathcal{T}$.
    \end{claim}
  
  Note that every point of $\mathcal{X}\setminus \mathcal{X}_j$,
  is the label of some descendant of a node $v$ in the $j$-level of $T$.
  Thus, by Lemma \ref{lem:ConvexHull}, every point of $\mathcal{X}\setminus \mathcal{X}_j$
  is in the convex hull of $R_k(v)$ for some node $v$ in the $j$-level of $T$.
  Let $\mathcal{A}$ be the region obtained from the convex hull of $\mathcal{X}$, 
  by removing the convex hull of $R_k(v)$ for each $v$ in the $j$-level of $T$.
  Note that the set of points of $\mathcal{X}$ that are in $\mathcal{A}$ is $\mathcal{X}_j$.

  Let $\triangle$ be the interior of a triangle of $\mathcal{T}$.
  By the nested almost convex set definition, there is one point of $\mathcal{X}$ in $\triangle$.
  As $\triangle$ is contained in $\mathcal{A}$, this point must be in $\mathcal{X}_j$.
  This point must also be in $R_j$, 
  since $\mathcal{A}$ without $\triangle$ (and its boundary) is convex.
  Thus, there is one point of $R_j$ in the interior of each triangle of $\mathcal{T}$.
  As the triangles of $\mathcal{T}$ are defined by consecutive points of $R'_{k}$, 
  each point of $R_j$ is in at most two triangles of $\mathcal{T}$.
  Thereby Claim~\ref{cl:LabelInTriangle} follows from $|\mathcal{T}|=|R'_{k}|=2|R_j|$. 
 
  Let $\triangle'$ be the intersection of the triangle determined by 
  $\previous[R_{j+1}(u)]$, $\first[R_{j+1}(u)]$ and $\last[R_{j+1}(u)]$, 
  with the triangle determined by $\first[R_{j+1}(u)]$, $\last[R_{j+1}(u)]$ and $\next[R_{j+1}(u)]$.
  Note that $\first[R_{j+1}(u)]$ and $\last[R_{j+1}(u)]$ are the labels of the children of $u$.
  By definition of $\psi'$, $x_u$ is in the interior of every triangle determined by 
  $\first[R_{j+1}(u)]$, $\last[R_{j+1}(u)]$ and every other point of $R_{j+1}$;
  thus $x_u$ is in $\triangle'$.
  By Claim~\ref{cl:LabelInTriangle},
  $x_u$ is in the interior of two triangles of $\mathcal{T}$,
  but there are only two triangles of $\mathcal{T}$ that intersect $\triangle'$;
  these are the triangles determined by 
  $\previous[R_k(u)]$, $\first[R_k(u)]$ and $\last[R_k(u)]$, 
  and the triangle determined by 
  $\first[R_k(u)]$, $\last[R_k(u)]$, $\next[R_k(u)]$.
 \end{proof}
  
\subsubsection*{\ref{thm:characterization-region} $\implies$ \ref{thm:characterization-order}}
  
In this part we assume that there is a labeling $\psi'$ of $\mathcal{X}$
that is nested, adoptable and well laid;
and we prove that $\psi'$ is an internal separation and an external separation.
  
  \begin{lem}
   $\psi'$ is an internal separation.\label{lem:internal-separation}
  \end{lem}
  \begin{proof}
   Let $u$ be a node of $T$ and recall that $u(l)$, $u(r)$ are the left and right children of $u$, respectively.
   We need to prove that every point in $\mathcal{X}/\mathcal{X}_u$ is to the left of every
   $(\overline{\mathcal{X}_{u(l)}},\overline{\mathcal{X}_{u(r)}})$-line.  
   
   Let $\ell$ be the directed segment from $\first[R_k(u(l))]$ to $\last[R_k(u(r))]$.
   By Lemma~\ref{lem:LabelInTriangle}, 
   each point in $\mathcal{X}/\mathcal{X}_u$ is in the interior of a triangle
   whose vertices are to the left of, or on $\ell$;
   thus every point in $\mathcal{X}/\mathcal{X}_u$ is to the left of $\ell$.
   By Lemma~\ref{lem:ConvexHull}, 
   every point in $\overline{\mathcal{X}_{u(l)}}\cup \overline{\mathcal{X}_{u(r)}}$ 
   is to the right of $\ell$.
   We claim that:   
   \begin{claim} \label{claim:internal-separation}
    No $(\overline{\mathcal{X}_{u(l)}},\overline{\mathcal{X}_{u(r)}})$-line intersects $\ell$.
   \end{claim}
   As the end points of $\ell$, $\first[R_k(u(l))]$ and $\last[R_k(u(r))]$,
   are in the boundary of the convex hull of $\mathcal{X}$;
   to prove that every point in $\mathcal{X}/\mathcal{X}_u$ is to the left of every
   $(\overline{\mathcal{X}_{u(l)}},\overline{\mathcal{X}_{u(r)}})$-line,
   it is enough to show Claim~\ref{claim:internal-separation}.
   
   Let $P_1$ be the polygonal chain that starts at $q_1:=\first[R_k(u(l))]$,
   follows the points of $R_k(u(l))$ in counterclockwise order, and ends at $q_2:=\last[R_k(u(l))]$.
   Similarly,  let $P_2$ be the polygonal chain that starts at $q_3:=\first[R_k(u(r))]$,
   follows the points of $R_k(u(r))$ in counterclockwise order, and ends at $q_4:=\last[R_k(u(r))]$.
   To prove Claim~\ref{claim:internal-separation} it is enough to show that
   every $(\overline{\mathcal{X}_{u(l)}},\overline{\mathcal{X}_{u(r)}})$-line intersects both $P_1$ and $P_2$.
   
   Let $q$ be the intersection point of the diagonals of the quadrilateral defined by $q_1$, $q_2$, $q_3$ and $q_4$.
   By Lemma~\ref{lem:ConvexHull} and Lemma~\ref{lem:LabelInTriangle},
   $\overline{\mathcal{X}_{u(l)}}$ is contained in the convex hull of $P_1\cup \lbrace q \rbrace$,
   and $\overline{\mathcal{X}_{u(r)}}$ is contained in the convex hull of $P_2\cup \lbrace q \rbrace$.
   Let $\ell'$ be an $(\overline{\mathcal{X}_{u(l)}},\overline{\mathcal{X}_{u(r)}})$-line.
   Note that the slope of $\ell'$,
   is in the range from the slope of the line define by $q_1$ and $q_3$,
   to the slope of the line define by $q_2$ and $q_4$, in counterclockwise order.
   Thus $\ell'$ intersects both $P_1$ and $P_2$.
   
  \end{proof}

  \begin{lem}
   $\psi'$ is an external separation. \label{lem:external-separation}
  \end{lem}
  \begin{proof} 
  
  Let $u$ be a node of $T$.
  We need to prove that every point in $\mathcal{X}/\overline{\mathcal{X}_u}$,
  is to the left of every 
  $(\overline{\mathcal{X}_{u(l)}},\lbrace x_u\rbrace)$-line and to the left of every
  $(\lbrace x_u\rbrace,\overline{\mathcal{X}_{u(r)}})$-line. 
  We prove that every point in $\mathcal{X}/\overline{\mathcal{X}_u}$  is to the left of every 
  $(\overline{\mathcal{X}_{u(l)}},\lbrace x_u\rbrace)$-line.
  That every point in $\mathcal{X}/\overline{\mathcal{X}_u}$  is to the left of every 
  $(\lbrace x_u\rbrace,\overline{\mathcal{X}_{u(r)}})$-line can be proven in a similar way.
  
  Let $P$ be the polygonal chain that starts at $\next[R_k(u)]$, 
  follows the points of $R_k$ in counterclockwise order,
  and ends at $\previous[R_k(u)]$.
  Note that, by Lemma \ref{lem:LabelInTriangle}, $\mathcal{X}/\overline{\mathcal{X}_u}$ is contained in the convex hull of $P$.
  Thus, to prove that every point in $\mathcal{X}/\overline{\mathcal{X}_u}$  is to the left of every 
  $(\overline{\mathcal{X}_{u(l)}},\lbrace x_u\rbrace)$-line, 
  it is enough to show that $x_u$ is to the right of the directed line 
  from $\last[R_k(u(l))]$ to $\next[R_k(u)]$. See Figure \ref{fig:external}.
  
  \begin{figure}[ht]
    \centering \hspace{-1.5 cm } 
    \begin{subfigure}[b]{0.3\textwidth}
	\includegraphics[scale=.48]{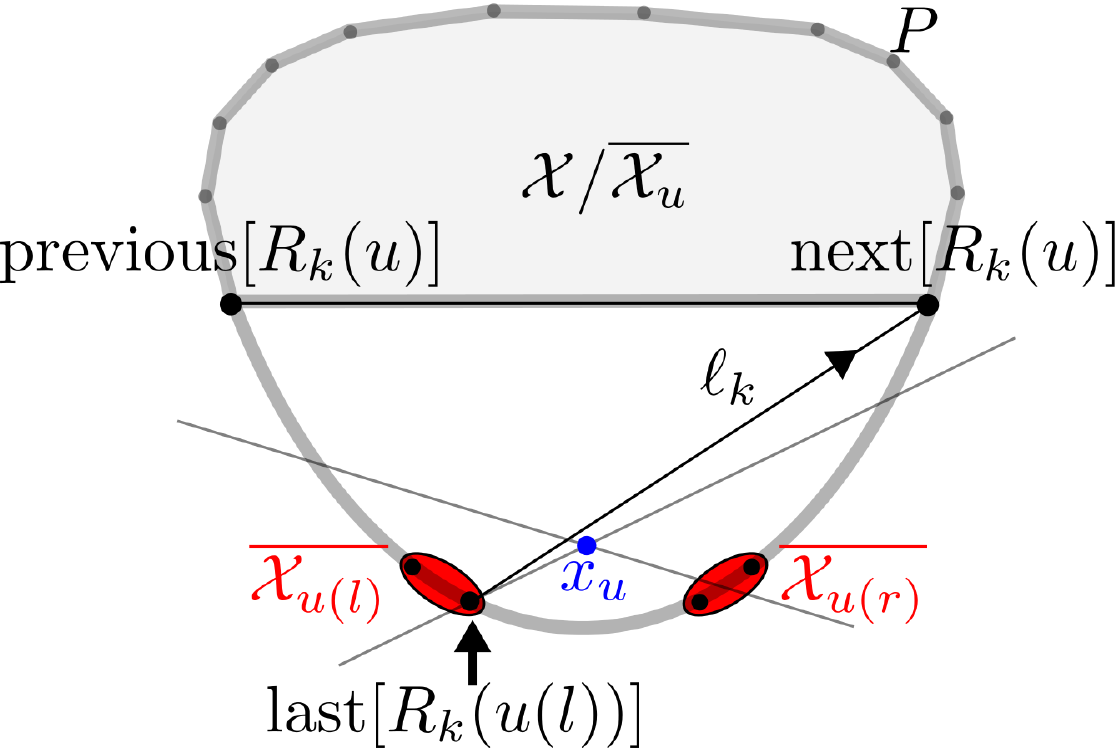}
	\label{fig:externala}
    \end{subfigure}
    \hspace{3 cm } 
    \begin{subfigure}[b]{0.3\textwidth}
	\includegraphics[scale=.38]{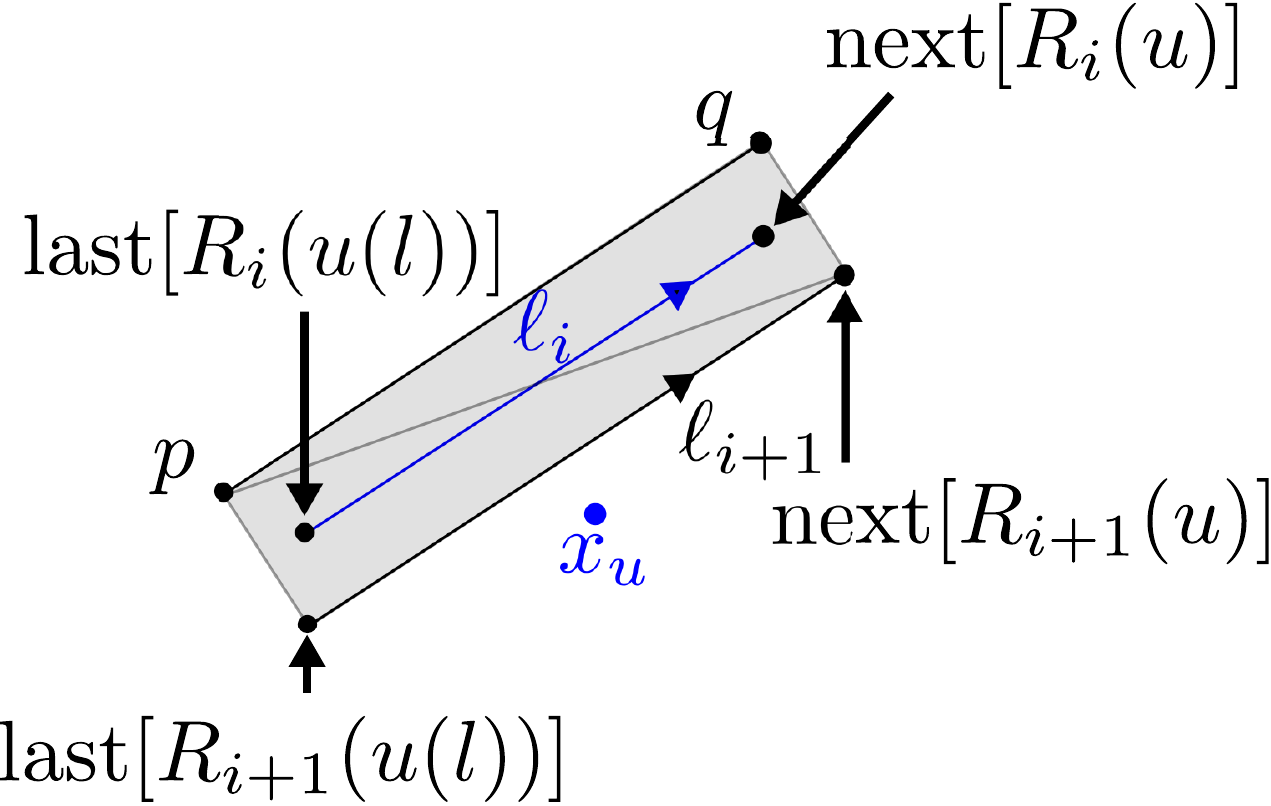}
	\label{fig:externalb}
    \end{subfigure}
    \caption{Illustration of the proof of Lemma~\ref{lem:external-separation}}\label{fig:external}
\end{figure}
  
  Let $j$ be the index such that the $j$-level of $T$ contains $u$. 
  For $j < i \leq k$, let $\ell_i$ be the directed line from $\last[R_i(u(l))]$ to $\next[R_i(u)]$.
  We show that $x_u$ is to the right of $\ell_i$ by induction. 
  As $x_{u(l)}$ and $x_{u(r)}$ are adoptable from $x_u$, and $x_{u(l)}=\last[R_{j+1}(u(l))]$;
  $x_u$ is in the interior of the triangle determined by 
  $\last[R_{j+1}(u(l))]$, $x_{u(r)}$ and $\next[R_{j+1}(u)]$.
  Thus the induction holds for $i=j+1$.
  Suppose that $x_u$ is to the right of $\ell_i$. 
  Let $\last[R_{i+1}(u(l))]$ and $p$ be the two children of $\last[R_{i}(u(l))]$.
  Let $\next[R_{i+1}(u)]$ and $q$ be the two children of $\next[R_{i}(u)]$.
  Let $\Box$ be the quadrilateral determined by $\last[R_{i+1}(u(l))]$, $p$, $q$ and $\next[R_{i+1}(u)]$.
  As $\last[R_{i+1}(u(l))]$, $p$, $q$ and $\next[R_{i+1}(u)]$ are in $R_{i+1}$, 
  and any triangulation of $\Box$ has two triangles;
  there are two points of $\mathcal{X}_i$ in $\Box$.
  As those points are $\last[R_i(u(l))]$ and $\next[R_i(u)]$, 
  $x_u$ is not in the interior of $\Box$. 
  Thus $x_u$ is not between $\ell_i$ and $\ell_{i+1}$, 
  and therefore $x_u$ is to the right of $\ell_{i+1}$.
  
  \end{proof}

\subsection*{\ref{thm:characterization-order} $\implies$ \ref{thm:characterization-def}}
In this part we finish the proof of Theorem~\ref{thm:characterization}.
We assume that there is a labeling $\psi'$ of $\mathcal{X}$
that is an internal separation and an external separation, 
and we prove that $\mathcal{X}$ is a nested almost convex set.
For this it is enough to prove Lemma~\ref{lem:DeterminesOrderTipe}.
As consequence of Lemma~\ref{lem:DeterminesOrderTipe} and Theorem~\ref{thm:characterization},
Theorem~\ref{thm:PossibleAlmosConvex} holds.

\begin{lem}\label{lem:DeterminesOrderTipe}
 Let $\mathcal{X}$ be an $n$-point set that admits a labeling $\psi:\mathcal{X} \rightarrow T$ 
 that is an internal separation and an external separation.
 Then the order type of $\mathcal{X}$ is is determined by $T$ and:
 \begin{itemize}
  \item If $n=2^{k-1}-2$, then $\mathcal{X}$ has the same order type than 
  any $n$-point set obtained from  Construction 1.
  \item If $n=3\cdot2^{k-1}-2$, then $\mathcal{X}$ has the same order type than 
  any $n$-point set obtained from  Construction 2.
 \end{itemize}
\end{lem}
\begin{proof}
The labeling that $\mathcal{X}$ admits can be a type 1 labeling or a type 2 labeling.
If $\mathcal{X}$ admits a type 1 labeling, $|\mathcal{X}|=2^{k+1}-2$ for some integer $k$;
in this case, an almost convex set with the same cardinality than $\mathcal{X}$
can be obtained using Construction 1. 
If $\mathcal{X}$ admits a type 2 labeling, $|\mathcal{X}|=3\cdot 2^{k-1}-2$ for some integer $k$;
in this case, an almost convex set with the same cardinality than $\mathcal{X}$
can be obtained using Construction 2. 
Let $\mathcal{Y}$ be an almost convex set with $|\mathcal{X}|$ points 
obtained from Construction 1 or Construction 2.
We prove that $\mathcal{X}$ and $\mathcal{Y}$ have the same order type,
and that this order type is determined by $T$.

Assume that $\mathcal{X}$ admits a type 1 labeling. 
The case when $\mathcal{X}$ admits a type~2 labeling can be proven in a similar way.
As $\mathcal{Y}$ is an almost convex set, $\mathcal{Y}$ admits a labeling that is an internal separation and an external separation.
Let $\psi_Y:\mathcal{Y}\rightarrow T$ be such type 1 labeling.

Let $f:=\psi^{-1}_Y(\psi')$.
We prove that $f:\mathcal{X}\rightarrow \mathcal{Y}$ is a bijection that preserves 
the orientation of all triplets. 
Let $x_1,x_2,x_3$ be different points in $\mathcal{X}$,
let $u_1,u_2,u_3$ be the nodes of $T$ that $x_1,x_2,x_3$ label in $\psi'$,
and let $y_1,y_2,y_3$ be the labels of $u_1,u_2,u_3$  in $\psi_Y$.
Note that $f(x_1)=y_1$, $f(x_2)=y_2$ and $f(x_3)=y_3$.
To prove that $(x_1,x_2,x_3)$ and $(y_1,y_2,y_3)$ have the same orientation,
we show that the position of $u_1$, $u_2$ and $u_3$ in $T$
determines the orientation of any labeling of $u_1$, $u_2$ and $u_3$.
 
 Given a node $w$ of $T$, denote by $T_w$ 
 the subtree of $T$ that contains every descendant of $w$.
 Let $w$ be the farthest node from the root of $T$, such that
 at least two of $u_1,u_2,u_3$ are in $T_w$.
 If two of $u_1,u_2,u_3$ are in the left subtree of $T_w$ or, 
 two of $u_1,u_2,u_3$ are in the right subtree of $T_w$;
 the orientation of the labels of $u_1,u_2,u_3$
 is determined by an external separation.
 If there are not two of $u_1,u_2,u_3$ in the left subtree of $T_w$ or in the right subtree of $T_w$;
 there is one of $u_1,u_2,u_3$ in the left subtree of $T_w$,
 one $u_1,u_2,u_3$ in the right subtree of $T_w$,
 and the other one is not in the left or right subtree of $T_w$. 
 In this case, the orientation of the labels of $u_1,u_2,u_3$
 is determined by an internal separation. 
 
\end{proof}

\section{Drawings of Nested Almost Convex Sets with Small Size.} \label{section:drawing}

Let $\mathcal{X}'$ be a nested almost convex set with $n$ points, 
and let $k$ be the number of convex layers of $\mathcal{X}'$.
In this section we construct a drawing of $\mathcal{X}'$ of size $O(n^{\log_2 5})$.
This section is divided into three parts.
First, we construct a $2^{k+1}-2$ point set $\mathcal{X}$ with integer coordinates and size $2\cdot5^{k+1}$.
Afterwards, we prove that $\mathcal{X}$ is a nested almost convex set.
Finally, we obtain a subset of $\mathcal{X}$ that is a drawing of $\mathcal{X}'$.

\subsection*{Construction of $\mathcal{X}$.}

Recall that  $T_1(k)$ is the complete binary tree  with $2^{k+1}-1$ nodes,
and the $j$-level of $T_1(k)$ is the set of nodes at distance $j$ from the root of $T_1(k)$.
Before defining $\mathcal{X}$, we will construct a point set $\mathcal{Y}$ in convex position,
and for each node $u$ in $T_1(k)$, 
we will define a set $\mathcal{Y}_u\subset \mathcal{Y}$ of consecutive points of $\mathcal{Y}$ in counterclockwise order.
The point $x_u$ will denote the midpoint between the first and last points of $\mathcal{Y}_u$ in counterclockwise order.
The set $\mathcal{X}$ will be the set of points $x_u$ such that $u$ is a node of $T_1(k)$ different from the root.

Let $p$, $o$ and $q$ be points in the plane and let $c\in[0,1]$.
We denote by $\overline{op}$ and $\overline{oq}$ the segments from $o$ to $p$ and from $o$ to $q$, respectively.
We say that $\alpha=(q,o,p)$ is a $\emph{corner}$,
if the angle from $\overline{op}$ to $\overline{oq}$ counterclockwise
is less than $\pi$.
Let $\alpha:=(q,o,p)$ be a corner.
We denote by $\Lp(\alpha,c)$ the point in the segment $\overline{oq}$ 
at distance $c|\overline{oq}|$ from $o$.
We denote by $\Rp(\alpha,c)$ the point in the segment $\overline{op}$ 
at distance $c|\overline{op}|$ from $o$. See Figure \ref{fig:corner}.

Recursively, we define a corner $\alpha_u$ for each node $u$ of $T_1(k)$.
The corner of the root of $T_1(k)$ is defined as $((0,2\cdot5^{k+1}),(0,0),(2\cdot5^{k+1},0))$.
Let $u$ be a node for which its corner $\alpha_u$ has been defined;
the corners of its left and right children, $u(l)$ and $u(r)$, 
are defined as follows (See Figure \ref{fig:corner}): 
\begin{align*} \label{eq:Corners}
\alpha_{u(l)} & = (\Lp(\alpha_u,2/5), \Lp(\alpha_u,1/5), \Rp(\alpha_u,1/5)) \\
\alpha_{u(r)} & =\left(\Lp(\alpha_u,1/5), \Rp(\alpha_u,1/5),  \Rp(\alpha_u,2/5)\right)
\end{align*}
 \begin{figure}[ht]
	\centering
	\includegraphics[scale=.21]{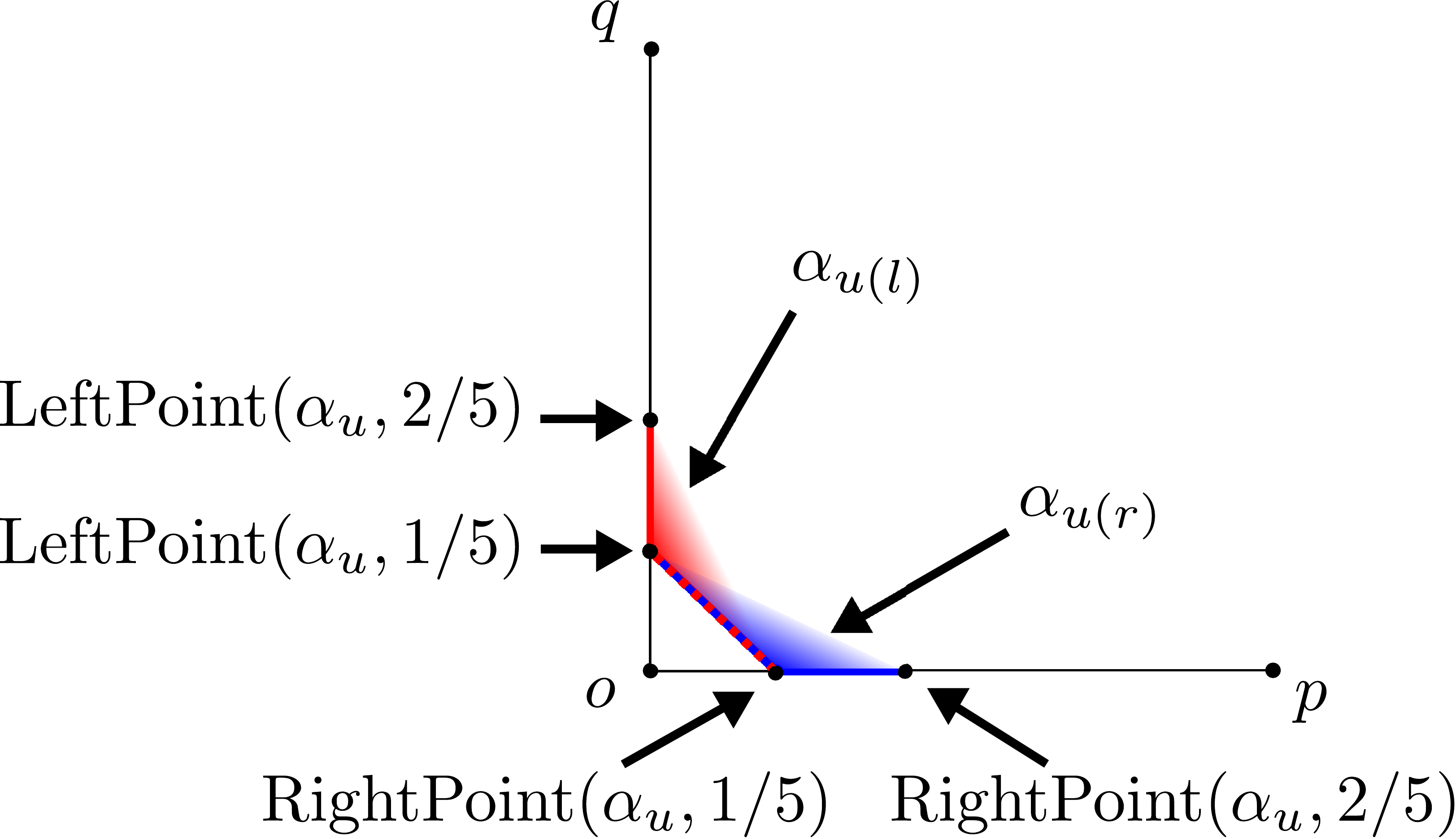}
	\caption{Illustration of corners $\alpha_u$, $\alpha_{u(l)}$ and $\alpha_{u(r)}$,
	where $\alpha_u=(q,o,p)$.}
	\label{fig:corner}
\end{figure}

Let $v$ be a leaf of $T_1(k+1)$.
Note that $v$ is a child of a leaf $u$ of $T_1(k)$.
If $v$ is the left child of $u$,
let $y_{v}:=\Lp(\alpha_u,1/5)$.
If $v$ is the right child of  $u$,
let $y_{v}:=\Rp(\alpha_u,1/5)$.
We define $\mathcal{Y}$ as the set of points $y_{v}$
such that $v$ is a leaf of $T_1(k+1)$.
Given a node $u$ of $T_1(k)$, 
we define $\mathcal{Y}_u$ as the set of points $y_{v}$
such that $v$ is a descendant of $u$, and $v$ is a leaf of $T_1(k+1)$.
With respect to the counterclockwise order, we denote by:
 $\first[\mathcal{Y}_u]$, the first point in $\mathcal{Y}_u$; 
 $\last[\mathcal{Y}_u]$, the last point in $\mathcal{Y}_u$;
 $\previous[\mathcal{Y}_u]$, the point in $\mathcal{Y}_u$ previous to $\first[\mathcal{Y}_u]$; and
 $\next[\mathcal{Y}_u]$, the point in $\mathcal{Y}_u$ next to $\last[\mathcal{Y}_u]$.

\begin{lem} \label{lem:Construction-properties1}
 Let $u$ be a node of $T_1(k)$.
 Let $v_1,v_2,\dots ,v_t$ be the leaves of $T_1(k+1)$, that are descendant of $u$,
 ordered from left to right.
 Then $y_{v_1}, y_{v_2},\dots,y_{v_t}$ are in convex position, 
 and are the points in $\mathcal{Y}_u$ in counterclockwise order.
\end{lem}
\begin{proof}
  Let $(q,o,p):=\alpha_u$; $q':=\Lp(\alpha_u,2/5)$; and $p':=\Rp(\alpha_u,2/5)$. 
  Let $\triangle(u)$ be the triangle determined by 
  $q'$, $o$ and $p'$.
  inductively from the leaves to the root of $T_1(k)$,
  it can be proven that:
  \begin{enumerate}
   \item The set of points of $\mathcal{Y}$ 
    in $\triangle(u)$ is $\mathcal{Y}_u$;
    from which:
    $\first[\mathcal{Y}_u]$ is on the segment from $o$ to $q'$,
    $\last[\mathcal{Y}_u]$ is on the segment from $o$ to $p'$,
    and the other points are in the interior of $\triangle(u)$.
   \item The points $q',y_{v_1}, y_{v_2},\dots,y_{v_t},p'$ are in convex position, 
    and appear in this order counterclockwise.
  \end{enumerate}
  This proof follows from 2.
\end{proof}

By Lemma~\ref{lem:Construction-properties1},
$\mathcal{Y}$ is in convex position,
and for each node $u$ in $T_1(k)$, 
$\mathcal{Y}_u$ is a subset of consecutive points of $\mathcal{Y}$ in counterclockwise order.
We denote by $x_u$ the midpoint between $\first[\mathcal{Y}_u]$ and $\last[\mathcal{Y}_u]$.
Let $\mathcal{X}$ be the set of points $x_u$ such that $u$ is a node of $T_1(k)$ different from the root.

Let $u$ be a node of $T_1(k)$ at distance $j$ from the root, 
let $(q,o,p):=\alpha_u$ and let $v$ be a leaf of $T_1(k+1)$.
Recursively  note that, the coordinates of $q$, $o$ and $p$ are divisible by $2\cdot5^{k+1-j}$.
Thus, the coordinates of $y_v$ are divisible by $2$,
$x_u$ has integer coordinates,
and $\mathcal{X}$ has size $2\cdot5^{k+1}$.

\subsection*{$\mathcal{X}$ is a nested Almost Convex Set.}
In this subsection we prove that $\mathcal{X}$ is a nested almost convex set. 
By Theorem~\ref{thm:characterization}, 
it is enough to prove that $\mathcal{X}$ admits a labeling that is an internal separation and an external separation.
Let $\psi:\mathcal{X}\rightarrow T_1(k)$ be the type 1 labeling that 
labels each node $u$ of $T_1(k)$ different from the root, with $x_u$.
We prove that $\psi$
is both an internal separation 
and an  external separation.

\begin{lem} \label{lem:Construction-propertiesC}
 If $u$ is a node of $T_1(k)$ at distance $j$ from the root,
 then $\first[\mathcal{Y}_{u}]=\Lp(\alpha_u,c_j)$ and 
 $\last[\mathcal{Y}_{u}]=\Rp(\alpha_u,c_j)$,
 where $$c_j=\frac{1}{4} \left( 1-5^{(j-k-1)}\right).$$
\end{lem}
\begin{proof}
 Note that \[c_j=\sum_{i=k}^j\left( \frac{1}{5}\right)^{k+1-j}.\]
 If $j=k$, then: $u$ is a leaf of $T_1(k)$; $c_j=1/5$;
 and $\first[\mathcal{Y}_{u}]=\Lp(\alpha_u,c_j)$ and 
 $\last[\mathcal{Y}_{u}]=\Rp(\alpha_u,c_j)$.
 Suppose that $j<k$, and that this lemma holds for larger values of $j$.
 Let $u(l)$ and $u(r)$ be the left and right children of $u$.
 Note that by induction, \[first[\mathcal{Y}_{u}]=\Lp(\alpha_{u(l)},c_{j+1})=\Lp(\alpha_{u},c*)\]
 where $c*=(1/5)c_{j+1}+1/5=c_j$; thus $\first[\mathcal{Y}_{u}]:=\Lp(\alpha_u,c_j)$.
 In a similar way $\last[\mathcal{Y}_{u}]:=\Rp(\alpha_u,c_j)$.
\end{proof}

 \begin{lem}
   $\psi$ is an internal separation.\label{lem:Construction:internal-separation}
  \end{lem}
  \begin{proof}
   Let $u$ be a node of $T_1(k)$ different from the root, 
   and let $u(l)$, $u(r)$ be the left and right children of $u$, respectively.
   We need to prove that every point in $\mathcal{X}/\mathcal{X}_u$ is to the left of every
   $(\overline{\mathcal{X}_{u(l)}},\overline{\mathcal{X}_{u(r)}})$-line.  
   
   Let $\ell$ be the directed segment from $\first[\mathcal{Y}_{u(l)}]$ to $\last[\mathcal{Y}_{u(r)}]$.
   As each point in $\mathcal{X}/\mathcal{X}_u$,
   is the midpoint between two points that are not to the right of $\ell$,
   every point in $\mathcal{X}/\mathcal{X}_u$ is not to the right of $\ell$.
   As every point in $\overline{\mathcal{X}_{u(l)}}\cup \overline{\mathcal{X}_{u(r)}}$,
   is the midpoint between a point to the right of $\ell$ and 
   a point that is not to the left of $\ell$, 
   every point in $\overline{\mathcal{X}_{u(l)}}\cup \overline{\mathcal{X}_{u(r)}}$ 
   is to the right of $\ell$. We claim that:
   
   \begin{claim} \label{claim:Construction:internal-separation}
    No $(\overline{\mathcal{X}_{u(l)}},\overline{\mathcal{X}_{u(r)}})$-line intersects $\ell$.
   \end{claim}
   
   As the endpoints of $\ell$, $\first[\mathcal{Y}_{u(l)}]$ and $\last[\mathcal{Y}_{u(r)}]$,
   are in the boundary of the convex hull of $\mathcal{Y}$;
   to prove that every point in $\mathcal{X}/\mathcal{X}_u$ is to the left of every
   $(\overline{\mathcal{X}_{u(l)}},\overline{\mathcal{X}_{u(r)}})$-line,
   it is enough to show Claim~\ref{claim:Construction:internal-separation}.
   
   Let $P_1$ be the polygonal chain that starts at $\first[\mathcal{Y}_{u(l)}]$,
   follows the points of $\mathcal{Y}_{u(l)}$ in counterclockwise order, and ends at $\last[\mathcal{Y}_{u(l)}]$.
   Similarly,  let $P_2$ be the polygonal chain that starts at $\first[\mathcal{Y}_{u(r)}]$,
   follows the points of $\mathcal{Y}_{u(r)}$ in counterclockwise order, and ends at $\last[\mathcal{Y}_{u(r)}]$.
   To prove Claim~\ref{claim:Construction:internal-separation} it is enough to show that
   every $(\overline{\mathcal{X}_{u(l)}},\overline{\mathcal{X}_{u(r)}})$-line intersects $P_1$ and $P_2$.
   This follows from the fact that $\overline{\mathcal{X}_{u(l)}}$ is contained in the convex hull of $P_1$,
   and $\overline{\mathcal{X}_{u(r)}}$ is contained in the convex hull of $P_2$.   
  \end{proof}

\begin{lem} \label{lem:Construction-properties2}
 Let $u$ be a node of $T_1(k)$ at distance $j$ from the root, and let $(q,o,p):=\alpha_u$.
 Suppose that the nodes in the $j$-level of $T_1(k)$,
 are ordered from left to right. 
 
 \begin{enumerate}
  \item \label{lem:Construction-properties2:collinear1}
    If $u$ is not the first node, then
    the points $o$, $\first[\mathcal{Y}_{u}]$, $\previous[\mathcal{Y}_{u}]$ and $q$ are collinear, 
    and $\previous[\mathcal{Y}_{u}]=\Lp(u,c)$ for some $c>3/5$.
  \item \label{lem:Construction-properties2:collinear2} 
    If $u$ is not the last node, then
    the points $o$, $\last[\mathcal{Y}_{u}]$, $\next[\mathcal{Y}_{u}]$ and $p$ are collinear,
    and $\next[\mathcal{Y}_{u}]=\Rp(u,c)$ for some $c>3/5$.

 \end{enumerate}
\end{lem}  
\begin{proof}
 To prove \ref{lem:Construction-properties2:collinear1} and \ref{lem:Construction-properties2:collinear2},
 note that, for any two consecutive nodes in the $j$-level of $T_1(k)$, 
 there is a segment that contains one side of each the corners corresponding to these nodes;
 then apply Lemma \ref{lem:Construction-propertiesC}.
\end{proof}

 \begin{lem}
   $\psi$ is an external separation.
  \end{lem}
  \begin{proof} 
  
  Let $u$ be a node of $T_1(k)$ and $u(l)$, $u(r)$ be the left and right children of $u$, respectively.
  We need to prove that every point in $\mathcal{X}/\overline{\mathcal{X}_u}$,
  is to the left of every 
  $(\overline{\mathcal{X}_{u(l)}},\lbrace x_u\rbrace)$-line and to the left of every
  $(\lbrace x_u\rbrace,\overline{\mathcal{X}_{u(r)}})$-line. 
  We prove that every point in $\mathcal{X}/\overline{\mathcal{X}_u}$  is to the left of every 
  $(\overline{\mathcal{X}_{u(l)}},\lbrace x_u\rbrace)$-line.
  That every point in $\mathcal{X}/\overline{\mathcal{X}_u}$  is to the left of every 
  $(\lbrace x_u\rbrace,\overline{\mathcal{X}_{u(r)}})$-line can be proven in a similar way.
  
  Let $P$ be the polygonal chain that starts at $\next[\mathcal{Y}_u]$, 
  follows the points of $\mathcal{Y}$ in counterclockwise order,
  and ends at $\previous[\mathcal{Y}_u]$.
  Note that $\mathcal{X}/\overline{\mathcal{X}_u}$ is contained in the convex hull of $P$.
  Thus, to prove that every point in $\mathcal{X}/\overline{\mathcal{X}_u}$  is to the left of every 
  $(\overline{\mathcal{X}_{u(l)}},\lbrace x_u\rbrace)$-line, 
  it is enough to show that $\next[\mathcal{Y}_u]$ is to the left of the directed line 
  from $\last[\mathcal{Y}_{u(l)}]$ to $x_u$.
  
  Let $\ell$ be the directed line from $\last[\mathcal{Y}_{u(l)}]$ to $x_u$ and let $(q,o,p):=\alpha_u$.
  Note that $x_u$ and $\last[\mathcal{Y}_{u(l)}]$ are in the interior of the wedge determined by $\alpha_u$,
  from $\overline{op}$ to $\overline{oq}$ in counterclockwise order.
  By Lemma \ref{lem:Construction-properties2}-\ref{lem:Construction-properties2:collinear2},
  $\next[\mathcal{Y}_u]$ is on $\overline{op}$ and $\next[\mathcal{Y}_{u}]=\Rp(u,c)$ for some $c>3/5$.
  To finish this proof we show that $\ell$ intersects $\overline{op}$
  at a point $\Rp(u,c')$ for some $c'<3/5$.
 
 Consider the following coordinate system, 
 $o$ is the origin,
 $p$ has coordinates $(1,0)$
 and $q$ has coordinates $(0,1)$.
 Assume that this is the new coordinate system.
 Let $t$ be such that the intersection point 
 between $\ell$ and the abscissa is the point $(t,0)$; 
 thereby, we need to prove that $t<3/5$.
 
 By Lemma \ref{lem:Construction-propertiesC},
 $\first[\mathcal{Y}_u]$ and $\last[\mathcal{Y}_u]$ have coordinates
 $(0,c_j)$ and $(c_j,0)$;
 thus, $x_u$ has coordinates $(c_j/2,c_j/2)$.
 By construction of $\alpha_{u(l)}$ %and $\alpha_{u(r)}$ 
 and Lemma \ref{lem:Construction-propertiesC},
 $\last[\mathcal{Y}_{u(l)}]$ is in the segment from $(0,1/5)$ to $(1/5,0)$
 in $\Rp(u(l),c_{j+1})$.
 Thus $\last[\mathcal{Y}_{u(l)}]$ has coordinates $(\frac{1}{5}c_{j+1},\frac{1}{5}(1-c_{j+1}))$ 
 and the equation of $\ell$ is 
 $$x=\frac{c_{j+1}/5-c_j/2}{(1-c_{j+1})/5-c_j/2}(y-c_j/2)+c_j/2 $$
 taking $y=0,s=k-j$, and replacing $c_j$ and $c_{j+1}$, we have that 
 $$t= -\frac{1}{40\cdot5^{s}} -\frac{1}{40(1+3/5^s)} - \frac{1}{40(3\cdot5^s + 5^{2s})}
 + \frac{3}{8(3/5^s + 1)} + \frac{1}{8(3+5^s)} + \frac{1}{8}$$
 finally, as $5^s \geq 1$ 
 \[t< \frac{3}{8} + \frac{1}{8(4)} + \frac{1}{8}=\frac{17}{32}<\frac{3}{5}.\]
\end{proof}

\subsection*{Construction of a Drawing of $\mathcal{X}$.}
 In this subsection we find a subset of $\mathcal{X}$ that is a drawing of $\mathcal{X}'$.
 By Theorem~\ref{thm:PossibleAlmosConvex}, there are two cases:
 $\mathcal{X}'$ is of type 1 and has $n=2^{k+1}-2$ points;
 or $\mathcal{X}'$ is of type 2 and has $n=3\cdot2^{k-1}-2$ points.
 By Theorem~\ref{thm:PossibleAlmosConvex},
 if $\mathcal{X}'$ is type 1, $\mathcal{X}'$ and $\mathcal{X}$ have the same order type
 and $\mathcal{X}$ is a drawing of $\mathcal{X}'$.
 Assume that $\mathcal{X}'$ is type~2.

 Let $w$ be the root of $T_1(k)$; $u$ and $u'$ be the children of $w$;
 $u(l)$ and $u(r)$ be the children of $u$; 
 and $u'(l)$ and $u'(r)$ be the children of $u'$.
 We define $T$ as the tree obtained from $T_1(k)$, 
 by making $u'(l)$ the third child of $u$ and 
 removing $w$, $u'$, $u'(r)$ and every descendant of $u'(r)$.
 Recall that $T_2(k)$ is a tree such that, its root has three children,
 and each child is the root of a complete binary tree with $2^{k-1}-1$ points.
 Note that $T$ and $T_2(k)$ are isomorphic.
 
 Let $\mathcal{X}_2$ be the set of points $x_u$ such that $u$ is in $T$. 
 Let $\psi':\mathcal{X}_2\rightarrow T$ be such that $\psi'(x_u)=u$.
 Note that: 
 as $\psi$ is an internal separation, $\psi'$ is an internal separation;
 and 
 as $\psi$ is an external separation, $\psi'$ is external separation.
 Thus by Theorem~\ref{thm:characterization}, $\mathcal{X}_2$ is a nested almost convex set.
 
 By Theorem \ref{thm:PossibleAlmosConvex}, 
 as $\mathcal{X}_2$ has $3\cdot2^{k-1}-2$ points, 
 $\mathcal{X}_2$ and $\mathcal{X}'$ have the same order type
 and $\mathcal{X}_2$ is a drawing of $\mathcal{X}'$.
 
\section{Decision Algorithm for Nested Almost Convexity.} \label{section:Certificate}

Let $\mathcal{X}$ be a set of $n$ points. 
In this section, we present an $O(n\log n)$ time algorithm,
to decide whether $\mathcal{X}$ is a nested almost convex set.
This algorithm is based in Theorem~\ref{thm:characterization}-\ref{thm:characterization-region} 
and consists of four steps. 
At each step, it is verified if $\mathcal{X}$ satisfies a certain property;
$\mathcal{X}$ is a nested almost convex set if and only if
$\mathcal{X}$ satisfies each of these properties.

By Theorem~\ref{thm:PossibleAlmosConvex}, 
if $\mathcal{X}$ is a nested almost convex set, then $n=2^{k-1}-2$ or $n=3\cdot2^{k-1}-2$
for some integer $k$.
The first step is to verify whether $\mathcal{X}$ has one of those cardinalities.
If $n=2^{k-1}-2$ let $T:=T_1(k)$.
If $n=3\cdot2^{k-1}-2$ let $T:=T_2(k)$.
Recall that:
the $j$-level of $T_1(k)$ is defined as the set of the nodes at distance $j$ from the root; and
the $j$-level of $T_2(k)$ is defined as the set of the nodes at distance $j-1$ from the root.

By Lemma~\ref{lem:ConvexLayers}, 
if $\mathcal{X}$ is a nested almost convex set then:
for $1\leq j\leq k$,
the number of nodes in the $j$-level of $T$
is equal to the number of nodes in the $j$-th convex layer of $\mathcal{X}$.
The second step is to verify whether $\mathcal{X}$ satisfies Lemma~\ref{lem:ConvexLayers}.
Chazelle \cite{Chazelle85ConvexLayers} showed that,
the convex layers of a given an $n$-point set can be found in $O(n\log n)$ time;
thus the second step can be done in $O(n\log n)$ time.
We denote by $R_j$ the set of points in the $j$-th convex layer of $\mathcal{X}$.

The third step is to verify whether $\mathcal{X}$ satisfies Lemma~\ref{lem:Adoptables}.
For $1\leq j \leq k-1$, we do the following.
Let $p_0,\dots p_{t}$ be the points in $R_j$ in counterclockwise order. 
We search for two consecutive points in $R_{j+1}$ that are adoptable by $p_0$.
If those points exist, they are the only pair of consecutive points in $R_{j+1}$ that are adoptable by $p_0$. 
Let $q_0,q_1,\dots p_{2t+1}$ be the points in $R_{j+1}$ in counterclockwise order, 
such that $q_0$ and $q_1$ are adoptable by $p_0$. 
Then we verify whether $q_{2i},q_{2i+1}$ are adoptable by $p_i$ for $0\leq i \leq t$.

Let $p$ be in $R_j$, and
let $q_r$, $q_{r+1}$, $q_{r+2}$, $q_{r+3}$ be four consecutive points in $R_{j+1}$.
Note that $q_{r+1}$ and $q_{r+2}$ are adoptable by $p$, if and only if, 
$p$ is in the intersection of the triangle determined by $q_r$, $q_{r+1}$ and $q_{r+2}$,
and the triangle determined by $q_{r+1}$, $q_{r+2}$ and $q_{r+3}$.
Thus, we can verify whether $q_{2i},q_{2i+1}$ are adoptable by $p_i$ in constant time;  
the third step hence requires linear time.

If $\mathcal{X}$ satisfies Lemma~\ref{lem:Adoptables}, we can
define a labeling $\psi:\mathcal{X}\rightarrow T$ like the one 
defined in Section \ref{Section:Characterization}-\ref{SubSection:AlmostConvexHasLabeling}.
The fourth step is to verify if $\psi$ is well laid, this requires linear time.

According to the proof of Theorem~\ref{thm:characterization}, 
$\mathcal{X}$ is a nested almost convex set if and only if 
$\mathcal{X}$ verifies the properties in previous four steps.  
This can be done in $O(n\log n)$ time.

\bibliographystyle{plain}
\bibliography{bib2}

\end{document}